\documentclass[12pt]{iopart}
\usepackage{graphicx}
\usepackage{dcolumn}
\usepackage{bm}
\usepackage{wasysym}
\usepackage{latexsym}
\usepackage{amsthm}

\newtheorem{thm}{Theorem}[section]

\theoremstyle{remark}

\theoremstyle{definition}

\begin{document}

\title{Edge local complementation for logical cluster states}
\date{\today}

\author{Jaewoo Joo}
\address{School of Physics and Astronomy, University of Leeds, Leeds LS2 9JT, UK}

\author{David L. Feder}
\address{Institute for Quantum Information Science, University of Calgary, Alberta T2N 1N4, Canada}

\begin{abstract}
A method is presented for the implementation of edge local complementation in 
graph states, based on the application of two Hadamard operations and a single 
controlled-phase (CZ) gate. As an application, we demonstrate an efficient 
scheme to construct a one-dimensional logical cluster state based on the 
five-qubit quantum error-correcting code, using a sequence of edge local 
complementations. A single physical CZ operation, together with local 
operations, is sufficient to create a logical CZ operation between two 
logical qubits. The same construction can be used to generate any encoded
graph state. This approach in concatenation may allow one to create a 
hierarchical quantum network for quantum information tasks.

\end{abstract}

\maketitle

\section{Introduction}

Multipartite entangled states are fundamental resources for quantum 
computation, with many mysteries yet to be understood~\cite{Horodecki}.
A particularly useful and interesting set of multipartite entangled states are
the so-called graph states~\cite{HEB+04}. These are quantum states associated 
with mathematical graphs, where vertices represent qubits in superposition
states and edges represent the maximally entangling controlled-phase ($CZ$) 
gates between them. Building complex graph states is a difficult task in 
practice (i.e. in experiments), because it requires the application of $CZ$ 
gates between arbitrary qubits; that said, considerable strides have been made
in recent years~\cite{expt}. It is nevertheless useful to consider the 
circumstances under which specific multipartite graph states can be 
constructed efficiently.

A class of particularly useful graph states are quantum error-correcting codes 
(QECCs). These are used to prevent quantum information leakage, since quantum
information is generically fragile against interactions with the 
environment~\cite{NC+01}. Standard QECCs can protect quantum information
against an arbitrary error on a single qubit. Several schemes of
measurement-based quantum computation with embedded quantum error correction
have been recently proposed but the structure of logical cluster states is
very complex~\cite{SDKO+07,JF+09,FY+10}. Very recently, a concatenation
scheme for a single logical qubit encoded in the five-qubit QECC (5QECC) has
been studied in the graph-state context~\cite{BCG+09}. While topological 
approaches to fault-tolerance in graph-state quantum computation yield
higher error thresholds~\cite{RH+07}, directly encoding the quantum 
information in QECC graphs might turn out to be more practical experimentally 
if efficient methods for constructing these states can be found.

We propose that multipartite graph states, which are useful for
constructing logical cluster states with 5QECC, can be efficiently
built by local Hadamard operations from simpler graph states. In
this paper, we prove that the mathematical operation called {\em
edge local complementation} (ELC) \cite{B+88}, which is defined by a
series of {\em local complementation} (LC) operations on a graph
\cite{HEB+04,HF+81}, is efficiently realizable in specific graph states
because it is equivalent to the action of local Hadamard operations.
From the mathematical point of view, LC transforms a given graph into 
another, with a different adjacency matrix; in practice, local complementation
of a given vertex complements the subgraph corresponding to its neighborhood.
From the quantum information point of view, LC corresponds to a set of local 
operations on a given graph state that therefore preserves any entanglement 
measure, yet describes a different graph state. Yet the cost of generating the 
new graph from a completely unentangled state would be significantly higher if 
the total number of edges is larger than in the original graph state. Our 
results indicate that the apparently complex nature of multipartite 5QECC 
states should not in itself be an impediment to their experimental generation, 
because they are in fact generically simple graphs under ELC.

This paper is organized as follows. We introduce the graph state notation 
in Section~\ref{Sec2}. The definition of edge local complementation and
its equivalence to Hadamard operations in graph states are discussed in 
Section~\ref{Sec3}. In Section \ref{Sec4}, we present the step-wise method of 
building one-dimensional (1D) logical cluster states. Finally, we summarize
our results with future research interests.

\section{Background}
\label{Sec2}

Let us begin with the definition of graphs and graph operations.
In graph theory, a graph $G=(V,E)$ is given by $N$ vertices 
$V=\{a_1,\ldots,a_N\}$ and edges $E$ corresponding to a linked line between 
two adjacent (neighboring) vertices. We only consider simple graphs with no 
self-loops and no multiple edges. If a vertex $c\in V$ is chosen in a graph, 
the other vertices are represented by its $n$ neighboring vertices 
${\mathcal N}(c)=\{b_1,\ldots,b_n\}\in V$ and outer vertices 
$V\backslash\{c,b_1,\ldots,b_n\}=\{o_1,\ldots o_{N-n-1}\}\in V$. The 
neighborhood of all of the vertices is defined by the adjacency matrix $A$, an 
$N\times N$ symmetric matrix with elements $A_{ij}=1$ iff $\{a_i,a_j\}\in E$.

All simple graphs correspond to a class of quantum states called {\em graph 
states}~\cite{HEB+04}, in which each vertex is represented by a qubit in a 
superposition state and an edge corresponds to the application of a maximally 
entangling gate.
Specifically, an $N$-qubit graph state $|G\rangle$ is defined as
\begin{equation}
|G\rangle=\bigotimes_{i\neq j}\left(CZ_{i,j}\right)^{A_{ij}}
|+\rangle^{\otimes N},
\end{equation}
where $|\pm \rangle$=$(|0\rangle\pm|1\rangle)/\sqrt{2}$ are the 
$\pm 1$-eigenstates of $X$ (here $\{X,Y,Z\}$ are the $2\times 2$ Pauli 
operators), and $CZ=\mbox{diag}(1,1,1,-1)$ is the controlled-phase 
(controlled-$Z$) gate acting between two qubits. Graph states can be defined 
in at least two equivalent ways, both of which will prove useful for our
purposes. Because the $CZ$ operations can be written as
\begin{eqnarray}
\label{eq:CZ01} CZ_{i,j} &=&
{1\over2} \left(
I_{i} I_{j} + I_{i}  Z_{j} + Z_{i}  I_{j} - Z_{i} Z_{j} \right)
\nonumber \\ 
&=& {1\over2} ( 1 + (-1)^{x_j} + (-1)^{x_i} - (-1)^{(x_i + x_j)})
=(-1)^{x_i\,x_j},
\end{eqnarray}
where $I_i$ is the identity matrix applied at site $i$, the graph state is 
given by a quadratic form of a Boolean function $p(x)$
\begin{eqnarray}
\label{eq:quadratic01} |G\rangle= {1\over \sqrt{2^N}}
(-1)^{p(x)}|x_{1}\cdots x_{N} \rangle,
\end{eqnarray}
where $x_i\in\{0,1\}$ and $p(x)=\sum_{i\neq j}A_{ij} x_i x_j$~\cite{JB+03}. 
Obviously, the value $x_ix_j=1$ iff $x_i=x_j=1$ (otherwise the value is 0), 
so $p(x)$ is a quadratic polynomial representing the graph adjacency matrix.
Alternatively, the state $|G\rangle$ is the fixed eigenvector, with unit 
eigenvalue, of the $N$ independent commuting operators 
\begin{equation}
S(a)=X_a\bigotimes_{b={\mathcal N}(a)}Z_b , 
\end{equation}
i.e.\ $S(a)|G\rangle=|G\rangle$ for all $a\in V$. Because the 
$\{S(a),\,a\in V\}$ generate a set of $2^N$ stabilizer operators $S$, these
$N$ generators uniquely define $|G\rangle$. 

\begin{figure}
\centering
\includegraphics[height=8cm,angle=-90]{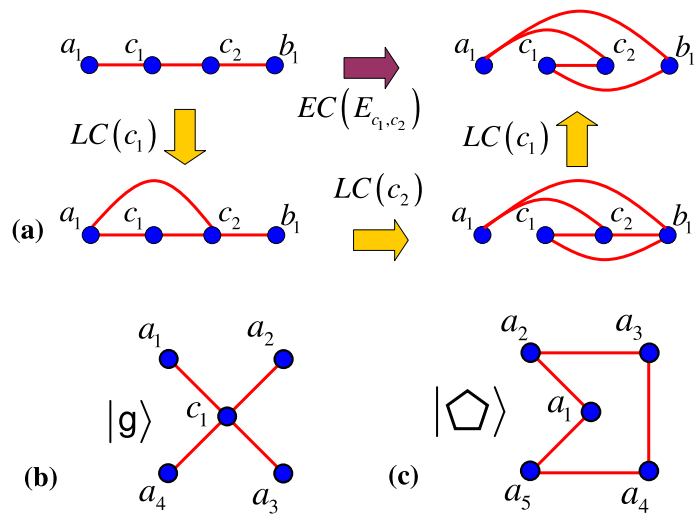}
\vspace{0.7cm} \caption{
\label{fig:edgecomplement01}
A simple example of edge local complementation on a graph state is shown in 
(a). The graph consists of four qubits (blue circles) and initially has three 
edges (red lines), forming a linear graph. Edge local complementation is given 
by three sequential (vertex) local complementations at $c_1$, $c_2$, and $c_1$, 
respectively. A five-qubit star graph state $|S_5\rangle\equiv|g\rangle$ is 
shown in (b) and (c) depicts a pentagon cycle graph state 
$|C_5\rangle\equiv|\pentagon\rangle$. 
}
\end{figure}

Fig.~\ref{fig:edgecomplement01}(b) and (c) shows two simple and important 
examples of graph states that are not equivalent under local-unitary
transformations, the star $S_N$ and cycle $C_N$ graphs. The star graphs 
correspond to GHZ states:
\begin{eqnarray}
\label{eq:GHZ}
|S_n\rangle = {1\over\sqrt{2^{n+1}}} (-1)^{q(x)}|x_{c_{1}}\,x_{a_1} 
\cdots x_{a_n} \rangle
\end{eqnarray}
with $q(x)=\sum_{i=1}^n x_{c_1} x_{a_i}$; these are LC-equivalent to the 
complete graphs $K_{n+1}$. Fig.~\ref{fig:edgecomplement01}(b) depicts 
$|S_4\rangle\equiv |g\rangle$. The cycle graph state is equal to
\begin{eqnarray}
\label{eq:Pentagon} |C_N\rangle
    = {1\over\sqrt{2^N}} (-1)^{r(x)} |x_{a_1}\cdots x_{a_N} \rangle
\end{eqnarray}
with $r(x)= x_{a_1} x_{a_N} + \sum_{i=2}^{N-1}x_{a_i} x_{a_{i+1}}$.
Fig.~\ref{fig:edgecomplement01}(c) shows $|C_5\rangle\equiv|\pentagon\rangle$.
The former is related to classically encoded graph states and the latter to 
5QECC \cite{Bennett+96}.

Local complementation and edge local complementation are two operations used 
to classify locally-equivalent graphs that are generally inequivalent under 
isomorphism (vertex permutation). The action of local complementation 
$LC(a)$ at the vertex $a$ transforms the graph $G$ by replacing the subgraph 
associated with the neighboring vertices ${\mathcal N}(a)$ by its 
complement~\cite{HF+81}. The new graph generated by $LC(a)$ on $G$ is locally 
equivalent to the original graph. It is important to note that the $LC(a)$ 
operation does not affect the edges of outer vertices in the graph $G$; only 
the neighborhood of vertex $a$ is affected. The action of edge local 
complementation $ELC(a,b)$ on the edge $\{a,b\}\in E$ is defined by three local 
complementations: $ELC(a,b)=LC(a)\,LC(b)\,LC(a)=LC(b)\,LC(a)\,LC(b)$.
The action of ELC on the edge $\{a,b\}$ can be understood as follows. Consider 
any pair of vertices $\{c,d\}\in E$, where $c$ is a neighbor of $a$ but not 
$b$ and $d$ is a neighbor of $b$ but not $a$ (or vice versa); alternatively, 
$c$ and $d$ can both be neighbors of $a$ and $b$. ELC then corresponds to 
complementing the edge between $c$ and $d$, i.e.\ if $\{c,d\}\in E$ then
delete the edge, and add it if $\{c,d\}\notin E$. In addition, the 
neighborhoods of $a$ and $b$ are replaced with one another. Edge local 
complementation has been investigated for recognizing the edge locally 
equivalence of two graphs~\cite{NestMoor+05} and for understanding the 
relationship between classical codes and graphs \cite{DPRK+10}.

In the context of graph states, local equivalence implies that one graph state 
can be transformed into another by the action of single-qubit (i.e.\ local)
operations. It is well-known that two graph states that are equivalent under 
stochastic local operations and classical communication (SLOCC) must also be 
equivalent under the local unitary (LU) operations~\cite{NDM+04}. A 
long-standing conjecture held that LU equivalence also implied
equivalence under the action of Clifford-group elements (operations that map
the Pauli group to itself), though this was recently proved to be 
false in general~\cite{JCWY+10}. 

Nevertheless, the transformations $LC$ (and 
therefore $ELC$) on graph states can be expressed solely in terms of local 
Clifford operations~\cite{HEB+04}:
\begin{equation}
LC(a)=\sqrt{-\imath X_a}\bigotimes_{b={\mathcal N}(a)}\sqrt{\imath Z_b}\propto
\sqrt{S(a)},
\label{eq:LCdef}
\end{equation}
where $\imath\equiv\sqrt{-1}$.
Suppose that $|G\rangle$ possesses qubit $c_1$ (called a {\em core} qubit) 
connected to $n$ neighboring qubits. The action of $LC(c_1)$ corresponds to the 
application of ${n(n-1)\over 2}$ $CZ$ operations on the graph state, creating 
edges between ${\mathcal N}(c_1)$ if there were none and removing them 
otherwise ($CZ^2=I$). Although entanglement between the two graph states is 
the same due to the invariance of entanglement under local unitary operations,
the number of {\it effective} $CZ$ operations (i.e.\ the number of edges) 
differs. Edge local complementation on the edge $\{a,b\}$ would then 
correspond to the operation
\begin{equation}
ELC(a,b)=\sqrt{-\imath X_a}\bigotimes_{c={\mathcal N}(a)}\sqrt{\imath Z_c}
\sqrt{-\imath X_b}\bigotimes_{d={\mathcal N}'(b)}\sqrt{\imath Z_d}
\sqrt{-\imath X_a}\bigotimes_{f={\mathcal N}''(a)}\sqrt{\imath Z_f},
\end{equation}
where the ${\mathcal N}'$ and ${\mathcal N}''$ are reminders that the 
neighborhoods themselves change under the $LC$ operations. Recognizing that
$a$ and $b$ remain neighbors, this can be rewritten
\begin{equation}
ELC(a,b)= (- \imath) H_a\otimes H_b \bigotimes_{c={\mathcal N}(a)\backslash b}
\sqrt{\imath Z_c}\bigotimes_{d={\mathcal N}'(b)\backslash a}\sqrt{\imath Z_d}
\bigotimes_{f={\mathcal N}''(a)\backslash b}\sqrt{\imath Z_f},
\end{equation}
where $H_a=(X_a+Z_a)/\sqrt{2} = \sqrt{-\imath X_a} \sqrt{\imath Z_a}\sqrt{-\imath X_a}$ is the Hadamard operator on qubit $a$. One of the goals of this 
manuscript is to show that the result of this operation on graph states can be 
expressed in the simpler form $ELC(a,b)|G\rangle=H_a\otimes H_b|G\rangle$, 
requiring the application of far fewer local operations.

Simple examples of LC and ELC are shown in 
Fig.~\ref{fig:edgecomplement01}(a). 
The initial graph state $|G\rangle$ consists of four qubits and three edges. 
After the first $LC(c_{1})$, because no edge exists between two neighboring
qubits of $c_{1}$ in state $|G\rangle$, an edge is drawn between
them. After $LC(c_{2})$, the edge on qubits $a_1$ and $c_1$ is
deleted by a rule of the local complementation because two
sequential $CZ$ operations become the identity between $a_1$ and
$c_1$. Finally, after the last $LC(c_1)$, the number of edges are
four on the final graph state, which is represented by $ELC(c_1,c_2)|G\rangle$, 
although all four graph states are locally equivalent.

\section{Edge local complementation via Hadamard gates}
\label{Sec3}

Consider two disconnected graphs $G_1=(V_1,E_1)$ and $G_2=(V_2,E_2)$ and their
respective graph states states $|G_1\rangle$ and $|G_2\rangle$; each possesses
a core vertex (qubit) $c_{1}\in V_1$ and $c_{2}\in V_2$, respectively. A $CZ$ 
operation is then applied to the two core qubits, linking the two graph states 
into a single connected graph 
$|G_{u}\rangle=CZ_{c_1,c_2}|G_1\rangle\otimes|G_2\rangle$. If 
a Hadamard operation is then applied to each core qubit, the graph 
$|G_u\rangle$ is transformed into another locally equivalent graph state 
$|G_H\rangle=H_{c_1}\otimes H_{c_2}|G_u\rangle$. Below we show that the state
$|G_{H}\rangle$ is the edge local complement of $|G_{u}\rangle$, i.e.\ that 
$|G_{H}\rangle = H_{c_1}\otimes H_{c_2}|G_u\rangle=ELC(c_1,c_2)|G_{u}\rangle$.
It is important to note that the equivalence of edge complementation on 
$\{c_1,c_2\}\in E$ with the application of Hadamard operations on $c_1$ and 
$c_2$ is only valid if ${\mathcal N}(c_1)\backslash c_2\cap {\mathcal N}(c_2)
\backslash c_1=0$, i.e.\ that prior to the application of $CZ_{c_1,c_2}$,
the neighborhoods of $c_1$ and $c_2$ were completely disjoint. Our results
do not apply to graphs where $c_1$ and $c_2$ share a neighborhood (other than 
themselves).

The main theorem of the paper is the following:

\begin{thm}
Consider two graph states 
$|G_1\rangle= {1\over \sqrt{2^{N_1}}} (-1)^{\sum_{i\neq j}A_{ij}^{(1)}x_ix_j}
|x_{1}\cdots x_{N_1} \rangle$ and $|G_2\rangle= {1\over \sqrt{2^{N_2}}}
(-1)^{\sum_{i\neq j}A_{ij}^{(2)}x_ix_j}|x_{1}\cdots x_{N_2} \rangle$,
defined by adjacency matrices $A^{(1)}$ and $A^{(2)}$ on independent vertex
sets $V_1\in\{a_1^{(1)},\ldots , a_{N_1}^{(1)}\}$ and 
$V_2\in\{a_1^{(2)},\ldots , a_{N_2}^{(2)}\}$, respectively. If core qubits, 
$c_1$ and $c_{2}$, are chosen at random from each of these vertex sets, and are 
entangled with one another by means of a $CZ$ gate, then
\begin{equation}
H_{c_1}H_{c_2}CZ_{c_1,c_2}|G_1\rangle |G_2\rangle=ELC(c_1,c_2)CZ_{c_1,c_2}
|G_1\rangle |G_2\rangle,
\end{equation}
where the edge local complementation operator on the edge $\{c_1,c_2\}$ is 
$ELC(c_1,c_2)=LC(c_1)LC(c_2)LC(c_1)$, and the (vertex) local complementation 
operator at qubit $a$ complements the edge set of its neighborhood 
${\mathcal N}(a)$.
\end{thm}

\begin{proof}
The core qubits $c_1$ and $c_{2}$ have neighborhood 
${\mathcal N}(c_1)=\{ b_{1}^{(1)},\ldots,b_{n}^{(1)}\}$ and
${\mathcal N}(c_2)=\{b_{1}^{(2)},\ldots, b_{m}^{(2)} \}$), respectively. The 
remaining vertices of the graphs $G_1\rangle$ and $G_2\rangle$ are 
$V_1\backslash\{c_1,b_1^{(1)},\ldots,b_n^{(1)}\}
=\{o_1^{(1)},\ldots o_{N_1-n-1}^{(1)}\}\in V_1$ and
$V_2\backslash\{c_2,b_1^{(2)},\ldots,b_m^{(2)}\}
=\{o_1^{(2)},\ldots o_{N_2-m-1}^{(2)}\}\in V_2$, respectively. Performing a 
$CZ$ operation between these core qubits, the graph state $|G_{u}\rangle$ is
\begin{eqnarray}
\label{eq:GHstep1}
|G_{u}\rangle
&=& CZ_{c_1,c_2} |G_{1}\rangle |G_{2}\rangle
= (-1)^{x_{c_1}x_{c_2}} |G_{1}\rangle |G_{2}\rangle, \nonumber \\
&=& \frac{(-1)^{x_{c_1}x_{c_2}}}{\sqrt{2^{N_1+N_2}}}(-1)^{\left(q_1(x)+q_2(x)
\right)}|x_{a_1^{(1)}} \cdots 
x_{a_{N_1}^{(1)}}\rangle|x_{a_1^{(2)}} \cdots x_{a_{N_2}^{(2)}}\rangle,
\end{eqnarray}
\begin{equation}
q_1(x)=\sum_{i\neq j}A_{ij}^{(1)}x_{a_i^{(1)}}x_{a_j^{(1)}}
=x_{c_1}\sum_{i=1}^nx_{b_i^{(1)}}
+\sum_{i\neq j}A_{ij}^{(1)}x_{o_i^{(1)}}x_{o_j^{(1)}};
\end{equation}
\begin{equation}
q_2(x)=\sum_{i\neq j}A_{ij}^{(2)}x_{a_i^{(2)}}x_{a_j^{(2)}}
=x_{c_2}\sum_{i=1}^mx_{b_i^{(2)}}
+\sum_{i\neq j}A_{ij}^{(2)}x_{o_i^{(2)}}x_{o_j^{(2)}}.
\end{equation}

\subsection{Two Hadamards applied to core qubits}

Consider
\begin{equation}
|G_{H}\rangle=H_{c_{1}}H_{c_{2}}|G_{u}\rangle=\frac{1}{2}\left(X_{c_{1}}X_{c_{2}}+X_{c_{1}}Z_{c_{2}}+Z_{c_{1}}X_{c_{2}}
+Z_{c_{1}}Z_{c_{2}}\right)|G_{u}\rangle.
\end{equation}
This can be simplified by noting that for $b_j\in{\mathcal N}(c_1)$
\begin{eqnarray}
X_{c_1} (-1)^{x_{c_1} x_{b_j}} &=& (-1)^{(x_{c_1} + 1) x_{b_j} } X_{c_1} 
;\nonumber \\
Z_{c_1} (-1)^{x_{c_1} x_{b_j}} &=& (-1)^{x_{c_1}}(-1)^{x_{c_1} x_{b_j}}
= (-1)^{x_{c_1}(x_{b_j}+1) }
.
\end{eqnarray}
One then obtains
\begin{eqnarray}
H_{c_{1}}H_{c_{2}}(-1)^{x_{c_1} x_{c_2}}&=&
\frac{1}{2} (-1)^{x_{c_1} x_{c_2}} \Big[ - (-1)^{(x_{c_1} + x_{c_2})}  X_{c_{1}} X_{c_{2}}\nonumber \\
&+& X_{c_{1}} + X_{c_{2}} + (-1)^{(x_{c_1}+x_{c_2})}\Big].
\label{eq:Hadamardstep01} 
\end{eqnarray}
Applying this to the remaining operators in Eq.~(\ref{eq:GHstep1}) gives
\begin{eqnarray}
|G_H\rangle &=& \frac{(-1)^{x_{c_1}x_{c_2}}}{\sqrt{2^{N_1+N_2}}}
(-1)^{\left(q_1(x)+q_2(x)\right)}\frac{1}{2}\prod_{k,k'}
\Big[ - (-1)^{(x_{c_1} + x_{c_2}+x_{b_k^{(1)}}+x_{b_{k'}^{(2)}})} 
\nonumber \\
&+& (-1)^{x_{b_k^{(1)}}} + (-1)^{x_{b_{k'}^{(2)}}} + (-1)^{(x_{c_1}+x_{c_2})}
\Big]|x_{a_1^{(1)}} \cdots x_{a_{N_1}^{(1)}}\rangle
|x_{a_1^{(2)}} \cdots x_{a_{N_2}^{(2)}}\rangle\nonumber \\
&=& \frac{(-1)^{x_{c_1}x_{c_2}}}{\sqrt{2^{N_1+N_2}}}
(-1)^{\left(q_1(x)+q_2(x)+q_3(x)\right)}
|x_{a_1^{(1)}} \cdots x_{a_{N_1}^{(1)}}\rangle
|x_{a_1^{(2)}} \cdots x_{a_{N_2}^{(2)}}\rangle,
\end{eqnarray}
where
\begin{equation}
q_3(x)=\left(x_{c_1}+x_{c_2}\right)\left(\sum_{k}x_{b_k^{(1)}}
+\sum_{k'}x_{b_{k'}^{(2)}}\right)+\sum_{k,k'}x_{b_k^{(1)}}x_{b_{k'}^{(2)}}.
\end{equation}
Finally, one can combine all the terms to obtain
\begin{eqnarray}
&& \hspace{-1cm} |G_H\rangle=H_{c_1}H_{c_2}CZ_{c1,c2}|G_1\rangle|G_2\rangle
=\frac{(-1)^{p(x)}}{\sqrt{2^{N_1+N_2}}}
|x_{a_1^{(1)}} \cdots x_{a_{N_1}^{(1)}}\rangle
|x_{a_1^{(2)}} \cdots x_{a_{N_2}^{(2)}}\rangle,
\label{eq:Hadamardfinal}
\end{eqnarray}
where 
\begin{eqnarray}
p(x)&=&x_{c_1}x_{c_2}+x_{c_1}\sum_{k'=1}^mx_{b_{k'}^{(2)}}
+x_{c_2}\sum_{k=1}^nx_{b_{k}^{(1)}}+\sum_{k=1}^n\sum_{k'=1}^mx_{b_k^{(1)}}
x_{b_{k'}^{(2)}}\nonumber \\
&+&\sum_{i\neq j}A_{ij}^{(1)}x_{o_i^{(1)}}x_{o_j^{(1)}}
+\sum_{i\neq j}A_{ij}^{(2)}x_{o_i^{(2)}}x_{o_j^{(2)}}.
\label{eq:quadraticfinal}
\end{eqnarray}

\subsection{Edge local complementation on core qubits}

Recall that edge local complementation $ELC(c_1,c_2)$ 
on the edge $\{c_1,c_2\}$ is described by the three local complementations 
$LC(c_1)LC(c_2)LC(c_1)=LC(c_2)LC(c_1)LC(c_2)$. Suppose that the first local 
complementation is performed on $|G_{u}\rangle$ at qubit $c_1$. The result is
that all neighboring qubits of $c_1$ are explicitly connected to each other 
(adding an edge to an existing edge annihilates it). The additional edges are 
given by the quadratic form 
\begin{equation}
r_1(x)=\sum_{i\neq j}x_{b_i^{(1)}}x_{b_j^{(1)}}+x_{c_2}\sum_ix_{b_i^{(1)}}.
\label{eq:r1}
\end{equation}
Next one complements the neighborhood of qubit $c_2$, which is given by the 
quadratic form $x_{c_2}\left(x_{c_1}+\sum_ix_{b_i^{(1)}}+\sum_ix_{b_i^{(2)}}
\right)$; the result is $x_{c_1}\left(\sum_ix_{b_i^{(1)}}+\sum_ix_{b_i^{(2)}}
\right)+\sum_{i\ne j}x_{b_i^{(1)}}x_{b_j^{(1)}}+\sum_{i\ne j}x_{b_i^{(2)}}
x_{b_j^{(2)}}+\sum_{i,j}x_{b_i^{(1)}}x_{b_j^{(2)}}$. The total additional 
edges are then given by the quadratic form
\begin{eqnarray}
r_2(x)&=&x_{c_2}\sum_ix_{b_i^{(1)}}+x_{c_1}\left(\sum_ix_{b_i^{(1)}}
+\sum_ix_{b_i^{(2)}}\right)+\sum_{i\ne j}x_{b_i^{(2)}}x_{b_j^{(2)}}\nonumber \\
&+&\sum_{i,j}x_{b_i^{(1)}}x_{b_j^{(2)}}.
\end{eqnarray}
Last, one complements the neighborhood of qubit $c_1$, which is given by the
quadratic form $x_{c_1}\left(x_{c_2}+\sum_ix_{b_i^{(1)}}+\sum_ix_{b_i^{(1)}}
+\sum_ix_{b_i^{(2)}}\right)=x_{c_1}\left(x_{c_2}+\sum_ix_{b_i^{(2)}}\right)$; 
the result is simply $x_{c_2}\sum_ix_{b_i^{(2)}}
+\sum_{i\neq j}x_{b_i^{(2)}}x_{b_j^{(2)}}$. The quadratic form for the 
additional edges after this final operation is 
\begin{eqnarray}
r_3(x)&=&x_{c_2}\left(\sum_ix_{b_i^{(1)}}+\sum_ix_{b_i^{(2)}}\right)
+x_{c_1}\left(\sum_ix_{b_i^{(1)}}+\sum_ix_{b_i^{(2)}}\right)\nonumber \\
&+&\sum_{i,j}x_{b_i^{(1)}}x_{b_j^{(2)}}.
\end{eqnarray}
Combining this result with the 
remaining terms in the quadratic form~(\ref{eq:GHstep1}), the graph resulting 
from the edge local complementation becomes
\begin{equation}
ELC(c_1,c_2)|G_u\rangle=\frac{(-1)^{p(x)}}{\sqrt{2^{N_1+N_2}}}
|x_{a_1^{(1)}} \cdots x_{a_{N_1}^{(1)}}\rangle
|x_{a_1^{(2)}} \cdots x_{a_{N_2}^{(2)}}\rangle,
\label{eq:ELCfinal}
\end{equation}
where
\begin{eqnarray}
p(x)&=&x_{c_1}x_{c_2}+x_{c_1}\sum_ix_{b_i^{(2)}}+x_{c_2}\sum_ix_{b_i^{(1)}}
+\sum_{i,j}x_{b_i^{(1)}}x_{b_j^{(2)}}\nonumber \\
&+&\sum_{i\neq j}A_{ij}^{(1)}x_{o_i^{(1)}}x_{o_j^{(1)}}
+\sum_{i\neq j}A_{ij}^{(2)}x_{o_i^{(2)}}x_{o_j^{(2)}},
\end{eqnarray}
which is identical to the quadratic form~(\ref{eq:quadraticfinal}).

\end{proof}

Eq.~(\ref{eq:quadraticfinal}) shows that when Hadamard gates are applied to 
both (core) qubits of single edge between two graphs, the result is a new graph 
state corresponding to the effective application of $2(n+m)+nm$ 
controlled-phase operations. These operations have the effect of replacing the 
original neighborhood of each core qubit with the neighborhood of the other
core qubit (and vice versa), while simultaneously adding the neighborhood of
a given core qubit to the neighborhood of the other. That is, from the edge
set one deletes the combinations $\{c_1,b_k^{(1)}\}$ and $\{c_2,b_k^{(2)}\}$,
and adds the combinations $\{c_1,b_k^{(2)}\}$, $\{c_2,b_k^{(1)}\}$, and
$\{b_k^{(1)},b_{k'}^{(2)}\}$. In other words, the Hadamard operations have 
complemented the neighborhood of the edge $\{c_1,c_2\}$, or performed edge 
local complementation. Of particular interest is the special case where both 
of the original graphs $|G_1\rangle$ and $|G_2\rangle$ were star graphs with 
the core qubit corresponding to the maximum-degree vertex, i.e.\ where
$\{o_i^{(1)}\}\not\in V_1$ and $\{o_i^{(2)}\}\not\in V_2$. Then the resulting 
graph would be completely bipartite, with every vertex of the first group 
$\{c_1,b_1^{(1)},\ldots,b_n^{(1)}\}$ connected to every vertex of the second 
group $\{c_2,b_1^{(2)},\ldots,b_n^{(2)}\}$~\cite{Gibbon}.

\subsection{Vertex local complementation}

The above analysis proves that the application of Hadamard operations to
the core qubits $c_1$ and $c_2$ is equivalent to edge local complementation
on the edge $\{c_1,c_2\}$. It is not obvious that edge local complementation
based on the formal definition of local complementation given in 
Eq.~(\ref{eq:LCdef}), $LC(c_1)=\sqrt{-\imath X_{c_1}}\prod_{b={\mathcal N}(c_1)}
\sqrt{\imath Z_b}$, reproduces the same result. Though graph transformations
effected by this expression have already been discussed in Ref.~\cite{HEB+04}
in the context of vertex local complementation, edge local complementation
using this operator was not explicitly explored in that work. In fact, as 
shown below, the application of these unitary gates in order to effect edge 
local complementation requires local operations in addition to the two 
Hadamard gates.

It is convenient to write 
\begin{equation}
\sqrt{-\imath X}=[-I+\imath X]/\sqrt{2};\quad
\sqrt{\imath Z}= [\imath I+ Z]/\sqrt{2}.
\end{equation}
The action of these on quadratic forms is
\begin{eqnarray}
\sqrt{-\imath X_{c_1}}(-1)^{x_{c_1}x_{b_j}}&=&(-1)^{x_{c_1}x_{b_j}} 
\left[-1+i(-1)^{x_{b_j}}X_{c_1} \right]/\sqrt{2}; \nonumber \\
\sqrt{\imath Z_{b_j}}(-1)^{x_{c_1}x_{b_j}}&=&  (-1)^{x_{c_1}x_{b_j}}  \left[\imath + (-1)^{x_{b_j}} \right]/\sqrt{2}.
\end{eqnarray}
Suppose one has an arbitrary graph state $|G\rangle$ defined by quadratic 
form $p(x)$ whose neighborhood of the qubit $c_1$ is 
${\mathcal N}(c_1)=\{b_1,\ldots,b_n\}$, i.e.\ where $p(x)$ includes the term 
$c_1\sum_jb_j$. Local complementation on the vertex $c_1$ then yields
\begin{eqnarray}
LC(c_1)&=&\sqrt{-\imath X_{c_1}} \prod_{j=1}^n 
\sqrt{\imath Z_{b_j}}(-1)^{x_{c_1}x_{b_j}}\nonumber \\
&=& \frac{1}{2^{n/2}} \sqrt{-\imath X_{c_1}} \prod_{j=1}^n 
(-1)^{x_{c_1}x_{b_j}}  \left[\imath + (-1)^{x_{b_j}} \right] \nonumber \\
&=& {1 \over 2^n} \prod_{j=1}^n (-1)^{x_{c_1}x_{b_j}}
\left[\imath + (-1)^{x_{b_j}} \right] \left[-1+i (-1)^{x_{b_j}}X_{c_1} \right]. \nonumber
\end{eqnarray}
When this local complementation operator is applied to the graph state 
$|G\rangle$, the $X_{c_1}$ operator will act only on its eigenstates and will
effectively disappear. The effect of the various terms above is then equivalent
to the new quadratic form 
\begin{equation}
p(x)'=p(x)+\sum_{j\neq k}b_j b_k - \sum_i b_i.
\end{equation}
In other words, $LC(c_1)$ has complemented the neighborhood of qubit $c_1$, by
effectively applying $CZ$ entangling operations to all of its neighbors. In 
addition, it has applied $Z$ gates to all the neighbors. These are local 
operations that commute with the $CZ$s and are therefore unimportant. That 
said, complete equivalence (rather than simply unitary equivalence) under 
edge local complementation would then require the application of additional 
unitary gates beyond the two Hadamard gates.

\section{Application : Efficient generation of 1D logical cluster states}
\label{Sec4}

We now discuss a novel and useful application of the theory of edge local
complementation for quantum information processing. In previous 
work~\cite{JF+09}, we showed that logical cluster states corresponding to 
5QECC can be made with logical $CZ$ operations consisting of many $CZ$ 
operations among the physical qubits. A linear $N$-qubit logical cluster state 
is given by
\begin{equation}
|CS^{L}_{N} \rangle = \prod_{i=1}^{N-1} CZ^{L}_{i\,i+1} |+^{L} \rangle_{1} 
\otimes \cdots \otimes |+^{L} \rangle_{N},
\end{equation}
where $CZ^{L}$ is a logical $CZ$ operation between two logical
qubits and $|\pm^{L} \rangle=(| 0^{L} \rangle \pm | 1^{L}
\rangle)/\sqrt{2}$. For $|CS^{L}_{2} \rangle$ with 5QECC, 25
physical $CZ$ operations are required to construct a logical $CZ$
operation from $| + \rangle^{\otimes 10}$ (see Fig.~3 in Ref.~\cite{JF+09}). 
The construction of many-qubit logical cluster states requires so many
entangling operations to build logical $CZ$ gates as to be impractical for
realistic quantum information processing. In this context, the edge local 
complementation provides an efficient solution to this conundrum: a single 
physical $CZ$ operation and two Hadamard operations are sufficient to build a 
logical $CZ$ operation between two logical qubits.

First we will review how to encode a physical qubit into a logical qubit with 
5QECC. One begins begin with a qubit in state $|0\rangle_{a_1}$ and four 
auxiliary qubits in $|++++\rangle_{a_2 - a_5}$. After a Hadamard operation on 
qubit $a_1$ and four $CZ$ operations between $a_1$ and the others, one obtains
the five-qubit GHZ-type graph state $|g\rangle_{A}$ (see
Fig.~\ref{fig:edgecomplement01}(b) but with $c_1$ replaced by $a_1$ and 
$a_{1,2,3,4}$ replaced by $a_{2,3,4,5}$). After an additional Hadamard 
operation on qubit $a_1$ in $|g\rangle_{A}$, the state is equal to a five-qubit
GHZ state
\begin{eqnarray}
\label{eq:Block02}
    |g^{+}_{H} \rangle_{A} = {1\over\sqrt{2}}
    \big(|+\rangle^{\otimes 5}_{a_1 - a_5} + |-\rangle^{\otimes
    5}_{a_1 - a_5}\big),
\end{eqnarray}
$|\pm\rangle^{\otimes 5}_{a_1 - a_5} = |\pm\rangle_{a_1}
|\pm\rangle_{a_2}|\pm\rangle_{a_3}|\pm\rangle_{a_4}|\pm\rangle_{a_5}$.
Because $|0\rangle_{a_1}=(|+\rangle_{a_1}+|-\rangle_{a_1})/\sqrt{2}$, the
state $|g^{+}_{H} \rangle$ can be understood as a classically
encoded state of $|0 \rangle_{a_1}$ in five qubits (here a classical encoding 
is meant to signify the implementation of a repetition code $|\pm\rangle
\rightarrow |\pm\rangle^{\otimes 5}$). Similarly, if the physical
qubit is initialized in $|1\rangle_{a_1}$, the outcome state is
$|g^{-}_{H} \rangle_{A} = \big(|+\rangle^{\otimes 5}_{a_1 - a_5}
- |-\rangle^{\otimes 5}_{a_1 - a_5}\big)/\sqrt{2}$. The quantum
encoding scheme transforms $|+\rangle^{\otimes 5}_{A}$ into 
$|+^L \rangle_{A}$ and $|-\rangle^{\otimes 5}_{A}$ into $|-^L \rangle_{A}$. 
As shown in Fig.~\ref{fig:edgecomplement01}(c), a pentagon graph operation is
used for encoding logical qubits
\begin{eqnarray}
| +^{L} \rangle_{A} &\equiv& |\pentagon\rangle_A
    = C^{\pentagon}_{a_1-a_5} | + \rangle^{\otimes 5}_{a_1-a_5},\label{eq:minusL} \\
| -^{L} \rangle_{A}&\equiv&|\tilde{\pentagon}\rangle_A
    =C^{\pentagon}_{a_1-a_5}| -\rangle^{\otimes 5}_{a_1-a_5}
    =\prod^{5}_{i=1} Z_i | +^{L} \rangle_{A}, \label{eq:plusL}
\end{eqnarray}
where $C^{\pentagon}_{a_1 - a_5} = CZ_{a_1,a_2}\, CZ_{a_2,a_3}\,
CZ_{a_3,a_4}\, CZ_{a_4,a_5}\, CZ_{a_5,a_1}$. Therefore, the total
encoding operation for a logical qubit is represented by
\begin{eqnarray}
\label{eq:ClassicE01}
    E^{L}_{A}= C^{\pentagon}_{a_1-a_5} \, H_{a_1} \, \left[\prod_{i=2}^{5}
    CZ_{a_1,a_i}\right] \, H_{a_1},
\end{eqnarray}
and $ | \pm^{L} \rangle_{A} = E^{L}_{A} | \pm \rangle_{a_1}| +
\rangle^{\otimes 4}_{a_2 - a_5}$.

\begin{figure}[t]
\centering
\includegraphics[height=8.5cm,angle=-90]{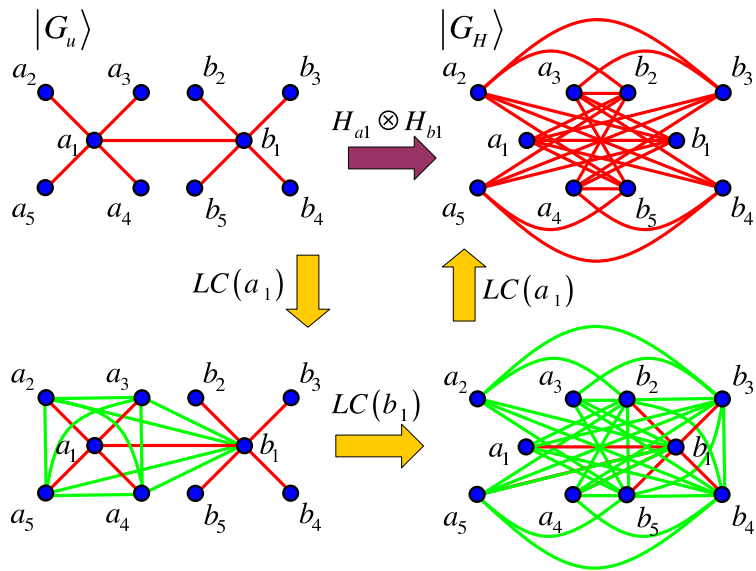}
\vspace{1cm} \caption{ \label{fig:Equ-example} The edge local
complementation consisting of three sequential local
complementations makes the same transformation as the operation
of two Hadamard operations on core qubits. Green lines indicate
new edges created by each local complementation while red lines
do the pre-existing edges. The graph state $|G_H\rangle$ is completely
bipartite.}
\end{figure}

With this toolkit one can show how to build logical cluster states. There are 
two different ways of building a two-qubit logical cluster state from
ten physical qubits $|+\rangle^{\otimes 10}$. The first method is to first 
prepare two logical states in $|+^{L}\rangle$, and then to directly perform a 
logical $CZ$ operation between them:
\begin{eqnarray}
|CS^{L}_{2}\rangle_{AB}  = CZ^{L}_{AB}  \, E^{L}_{A} \, E^{L}_{B}
|+\rangle^{\otimes 10} = CZ^{L}_{AB} | +^{L} \rangle_{A} | +^{L}
\rangle_{B}.
\end{eqnarray}
Because $E^{L}_{A} E^{L}_{B} |+\rangle^{\otimes 10} =
C^{\pentagon}_{A} C^{\pentagon}_{B}|+\rangle^{\otimes 10}$ for this case, 35 
physical $CZ$ operations in total are required to build $|CS^{L}_{2}\rangle$ 
from $| + \rangle^{\otimes 10}$ (for details refer to Ref.~\cite{JF+09}). 

In the second method, one creates classically encoded graph states by means of 
edge local complementations; the quantum encoding is then applied to the 
classically encoded states to obtain logical cluster states.  Initially, the
core qubit $a_1$ of one classical state is entangled with its counterpart $b_1$
in the other state, yielding a two-qubit cluster state 
$|CS_{2} \rangle_{a_1 b_1} =(|0 \rangle_{a_1} |+
\rangle_{b_1} + |1 \rangle_{a_1} |- \rangle_{b_2})/\sqrt{2}$.
Note that the first Hadamard operations $(H_{a_1} \otimes
H_{b_1})$ in $E^{L}_{A} E^{L}_{B}$ leave the state $|CS_{2} \rangle_{a_1 b_1}$ 
invariant. After the GHZ-type $CZ$ operations are performed between $a_1$
($b_1$) and $a_i$ ($b_i$) ($i=2,3,4,5$), through the operation 
$\prod_{i}\prod_{j} CZ_{a_1 , a_i} CZ_{b_1 , b_j}$, a connected graph state
$|G_u \rangle$ is obtained (see Fig.~\ref{fig:Equ-example}). When
two Hadamard operations are subsequently applied to $a_1$ and $b_1$ in $|G_u
\rangle$, the resulting state is transformed 
to another graph state $|G_H\rangle$, given by
\begin{eqnarray}
\label{eq:stateGH} \hspace{-1.5cm} |G_H\rangle_{AB} = {1\over2}
\left[ |+\rangle^{\otimes 5}_{a_1-a_5}\left(|+\rangle^{\otimes
5}_{b_1-b_5} + |-\rangle^{\otimes 5}_{b_1-b_5} \right) +
|-\rangle^{\otimes 5}_{a_1-a_5} \left(|+\rangle^{\otimes 5}_{b_1-b_5}
- |-\rangle^{\otimes 5}_{b_1-b_5} \right) \right].
\end{eqnarray}
This state is a classically encoded two-qubit cluster state.

In Fig.~\ref{fig:Equ-example}, it is shown that the action of
three local complementations on the core vertices $a_1$ and $b_1$ provides 
the desired $CZ$ operations among the physical qubits, reproducing the
state~(\ref{eq:stateGH}). The resulting graph is known as a complete bipartite 
graph state~\cite{Gibbon}: each of the vertices in one neighborhood 
(corresponding to logical register A or B) is connected with all the vertices
of the other neighborhood, and vice versa. While it is possible to construct
$|G_H\rangle_{AB}$ directly by applying 25 $CZ$ operations starting with
$|+\rangle^{\otimes 10}$, it can be efficiently made using only 9 $CZ$ 
operations plus two local operations. For the quantum encoding scheme, the 
final state is given by
\begin{eqnarray}
\label{eq:1st2CS} |CS^{L}_{2}\rangle_{AB}  =
C^{\pentagon}_{a_1-a_5}C^{\pentagon}_{b_1-b_5}|G_H\rangle_{AB} =
E^{L}_{A} E^{L}_{B} CZ_{a_1 b_1} |+\rangle^{\otimes 10}.
\end{eqnarray}
Therefore, the state $|CS^{L}_{2}\rangle$ can be efficiently built
by 19 $CZ$ operations with the help of two Hadamard operations,
instead of 35 $CZ$ operations, and the logical $CZ$ operation
expressed by
\begin{eqnarray}
\label{eq:logicalCZ} CZ^{L}_{AB} \, E^{L}_{A} \, E^{L}_{B} =
E^{L}_{A} \, E^{L}_{B} \, CZ_{a_1 b_1}
\end{eqnarray}
shows that a single physical $CZ$ operation is sufficient to
create a logical $CZ$ operation between logical qubits.

While the encoding procedure for graph states is straightforward to implement,
its interpretation in terms of edge local complementation is not obvious in
general. For example, any encoding of a cluster state with an odd number of 
qubits is difficult to express in terms of edge local complementations, each 
requiring an even number of Hadamard operations. The interpretation of 
encoding linear $2N$-qubit cluster states through edge local complementation is 
straightforward, however. 

Consider for example the
linear four-qubit logical cluster state.  First one assigns five qubits each
to registers A, B, C, and D. After assigning a core qubit from each, designated 
$a_1$, $b_1$, $c_1$, and $d_1$, respectively, one prepares the linear
four-qubit cluster state $|CS_{4}\rangle=CZ_{a_1,b_1}CZ_{b_1,c_1}CZ_{c_1,d_1}
|+\rangle^{\otimes 4}_{a_1-d_1}$. The encoding consists of acting on each 
register with $\prod_{J} E^{L}_{J}|CS_{4}\rangle|+\rangle^{\otimes 16}$ for 
$J=A,B,C,D$, where $E^{L}_{J}$ is given in Eq.~(\ref{eq:ClassicE01}). The first 
step is to perform four Hadamard operations on $|CS_{4}\rangle$. Applying two Hadamard operations on qubits $a_1$ 
and $b_1$,
the intermediate graph state is equal to
\begin{eqnarray}
\hspace{-1cm} |\Psi_{inter} \rangle_{a_1-d_1} = ELC(a_1,b_1)
|CS_{4}\rangle = CZ_{a_1,b_1}CZ_{a_1,c_1} CZ_{c_1,d_1}
|+\rangle^{\otimes 4}_{a_1 - d_1}
\end{eqnarray}
using the results of edge local complementation. Because $c_1$ and $d_1$
share an edge but their neighborhoods are disjoint, it is reasonable to 
associate the subsequent Hadamard operations on qubits $c_1$ and 
$d_1$ with another edge local complementation on the edge $\{c_1,d_1\}$. The 
resulting state is equal to another linear four-qubit cluster state, but with
the vertex labels permuted:
\begin{equation}
|CS_{4}' \rangle = ELC(c_1,d_1) |\Psi_{inter} \rangle_{a_1-d_1}
= CZ_{a_1,b_1}CZ_{a_1,d_1}CZ_{c_1,d_1}|+\rangle^{\otimes 4}_{a_1 - d_1}.
\end{equation}
After four GHZ-type operations and the second set of four Hadamard
operations on $|CS_{4}' \rangle_{a_1-d_1} |+\rangle^{\otimes
16}$, again corresponding to two edge local complementations, the
outcome is a linear four-qubit cluster state with classical encoding. The
Hadamard operations not only effect the edge local complementation; they also
reverse the permutation of the vertex labels above. Finally, the quantum 
encoding scheme on all the qubits yields a logical four-qubit cluster state 
$|CS^{L}_{4} \rangle$, which is sufficient for universal quantum
computation with 5QECC~\cite{JF+09}. This procedure can be trivially extended
to any even-length chain, by applying Hadamard gates in pairs on 
nearest-neighbor edges in order to implement edge local complementations from 
the left boundary of the chain to the right.


\section{Summary and Remarks}

The main result presented in this manuscript is a proof that the action of 
edge local complementation on a graph state can be effected solely through the 
use of two Hadamard operations applied to the edge qubits. A crucial 
assumption in this proof is that the neighborhoods of the edge qubits were 
disjoint, i.e.\ that the neighbors ${\mathcal N}(a)$ of the first edge qubit 
$a$ were different from the neighbors ${\mathcal N}(b)$ of the second qubit 
$b$. Under this restriction, edge local complementation interchanges the 
respective neighborhoods, i.e.\ ${\mathcal N}(a)\leftrightarrow{\mathcal N}(b)$,
while simultaneously making neighbors of all the neighbors. In principle, this
transformation would require a large number of either local unitary operations 
on the graph-state qubits or entangling gates between various qubits. The
distinct advantage of the present scheme is the large savings in the number
of (local) operations required.

As an example of the utility of this insight, we show how edge local
complementation can be used to efficiently create classically encoded cluster
states and one-dimensional logical cluster states based on the five-qubit
error-correcting code, for an even number of logical qubits. In this scheme, 
a physical $CZ$ operation, together with local operations, is sufficient to 
create a logical $CZ$ operation between two logical qubits.

Arbitrary encoded graph states can be obtained by a straightforward extension 
of the procedure described above. The operations encoding a logical qubit, 
Eq.~(\ref{eq:ClassicE01}), are local to the physical qubits comprising the 
logical qubit, and therefore commute with one another. It therefore suffices to
first construct the desired graph state with the core qubits, associate four 
ancillae to each core qubit, and operate independently with 
Eq.~(\ref{eq:ClassicE01}) on each five-qubit register.

Multipartite entangled states that fundamentally include fault tolerance might 
be desirable for practical measurement-based quantum computing and 
multipartite quantum communication~\cite{BCG+09,Knill}. For a generalized 
scheme of level-$l$ logical graph states based on our proposal, the same 
encoding procedure can be used repeatedly. Since the level-1 logical
graph state is made by our protocol, a level-$l$ concatenated
logical graph state becomes an initial state to create a
level-($l$+1) concatenated one ($|\pm^{L}_{l}\rangle \rightarrow
|\pm^{L}_{l+1}\rangle$) with the help of the classical and quantum
coding schemes described above. This concatenated method may also be useful for
building multi-party quantum networks similar to classical complex networks in 
hierarchical organization~\cite{E+A03}.

\section*{Acknowledgements}
The authors are grateful to T. P. Spiller and Jinhyoung Lee for stimulating
discussions. This work was supported by the Natural Sciences and Engineering 
Research Council of Canada, the Mathematics of Information Technology and
Complex Systems Quantum Information Processing Project, and the Quantum 
Interfaces, Sensors, and Communication based on Entanglement Integrating 
Project.

\end{document}